\documentclass{article}
\usepackage{graphicx} 
\usepackage{fullpage}

\newcommand{\alg}{\mathrm{ALG}}

\usepackage[]{color-edits}%
\addauthor{TE}{magenta}

\newcommand{\SET}[2]{\mathcal{S}(#1,#2)}
\newcommand{\oc}[2]{\#_{#1}(#2)}
\usepackage{amsmath}
\usepackage{amssymb}
\usepackage{mathtools}
\usepackage{amsthm}

\newtheorem{lemma}{Lemma}[section]

\newtheorem{theorem}{Theorem}[section]

\newtheorem{definition}{Definition}[section]

\newcommand{\event}[1]{\mathcal{E}_{#1}}
\newcommand{\pre}{\mathcal{P}}
\newcommand{\E}{\mathbb{E}}
\newcommand{\ps}{\mathcal{V}}
\newcommand{\ber}[1]{{Ber}{(#1)}}
\usepackage{natbib}

\newcommand{\bin}[2]{\mathcal{B}(#1,#2)}
\newcommand{\norm}[2]{\mathcal{N}(#1,#2)}

\title{Prophet Inequality from Samples: Is the More the Merrier?\footnote{T.~Ezra is supported by the Harvard University Center of Mathematical Sciences and Applications.}}
\author{Tomer Ezra\thanks{Harvard University, Cambridge, USA. Email: \texttt{tomer@cmsa.fas.harvard.edu}}}
\date{}

\begin{document}

\maketitle
\begin{abstract}

We study a variant of the single-choice prophet inequality problem where the decision-maker does not know the underlying distribution and has only access to a set of samples from the distributions. \citet{rubinstein2020optimal} showed that the optimal competitive-ratio of $\frac{1}{2}$ can surprisingly be obtained by observing a set of $n$ samples, one from each of the distributions. 
In this paper, we prove that this competitive-ratio of $\frac{1}{2}$ becomes unattainable when the decision-maker is provided with a set of more samples. 
We then examine the natural class of ordinal static threshold algorithms, where the algorithm selects the $i$-th highest ranked sample, sets this sample as a static threshold, and then chooses the first value that exceeds this threshold. We show that the best possible algorithm within this class achieves a competitive-ratio of $0.433$.
Along the way, we utilize the tools developed in the paper and provide an alternative proof of the main result of \citet{rubinstein2020optimal}.

\end{abstract}
\thispagestyle{empty} 
\newpage
\pagenumbering{arabic}
\section{Introduction}

In the prophet inequality setting \citep{krengel1977semiamarts,krengel1978semiamarts} a decision-maker needs to select a reward among a sequence of rewards in an online manner. The rewards are stochastic and are distributed according to known distributions. The performance of the decision-maker is measured against the offline optimum (or a ``prophet'' that can see into the future, and always selects the maximum reward). 

Traditionally, it is assumed that the distributions of the random variables are known to the decision-maker. This assumption, while convenient for theoretical analysis, raises practical concerns: Where do these priors come from, and how realistic is it to assume that they are known with precision? In many real-world scenarios, exact knowledge of the underlying distributions is not available, and instead, decision-makers must rely on samples drawn from these distributions. For instance, a store seller might document the selling prices of items without keeping detailed records of which customers made the purchases at those prices. Additionally, due to regulatory requirements, market platforms may be limited to retaining anonymized statistical data on past transactions, without associating specific prices with individual customers.

\citet{azar2014prophet} proposed addressing this challenge by considering the prophet inequality in a setting where only samples, rather than complete knowledge of the distributions, are available. This approach by \citet{azar2014prophet} sparked a line of research that explores the prophet inequality framework under various conditions where the online algorithm has access only to samples \citep{azar2014prophet,rubinstein2020optimal,dutting2021prophet,caramanis2022single,correa2022two,gravin2022optimal,kaplan2022online,cristi2024prophet,correa2024sample,fu2024sample}.

Among this line of works, \citet{rubinstein2020optimal} considered the classical setting of single-choice prophet inequality from samples and showed that it is possible to achieve a competitive-ratio of $\frac{1}{2}$ even when only a single-sample is provided from each distribution. Moreover, they demonstrated that this can be accomplished without knowing which sample is associated with each distribution.

This naturally leads to the question of whether the result extends beyond the single-sample case, which motivates our study: 
    \emph{Can an algorithm that receives a set of samples, with $ k $ samples from each distribution, 
 guarantee the optimal competitive-ratio of $\frac{1}{2}$?}

\subsection{Our Results and Techniques}
Our first main result is answering this question in the negative. 
\paragraph{Main Result 1 (Section~\ref{sec:main}):} \textit{There exists a constant $c < 1/2 $ such that no online algorithm receiving $k\cdot n $ samples can guarantee a competitive-ratio greater than $ c$. }

Our approach involves constructing a family of instances where any online algorithm will fail to achieve the desired competitive-ratio for at least one instance within the family. The instances we construct have only a constant number of rewards, with a constant number of values in their support. This allows us to characterize the behavior of an online algorithm as a function that maps any realized sample set and any prefix of values to the probability of accepting the last observed value.

Using this characterization, we establish that for two different instances, the difference in the probability of selecting the last reward when observing the same prefix can be bounded by the distance between the distributions over the sample sets. Moreover, this holds true even for convex combinations of such instances. Specifically, if the distribution over the sample set induced by one instance is close to a convex combination of distributions from other instances, then the algorithm should behave similarly to the convex combination of the behaviors of the algorithm in those instances. This insight allows us to construct several instances where, if the algorithm performs well on all of them, it must necessarily perform poorly on carefully constructed convex combinations of these instances.

To complement the hardness result, we analyze a specific class of algorithms we term as ordinal static threshold algorithms. These algorithms select a static threshold of a predefined rank among the samples, independent of the samples' values, and set it as a threshold. Then, when observing the online rewards, they select the first reward exceeding the threshold. 

\paragraph{Main Result 2 (Section~\ref{sec:ordinal}):} \textit{There exists an ordinal static threshold that achieves a competitive-ratio of approximately $0.433$, and this is the optimal competitive-ratio achievable by any ordinal static threshold algorithm.}

Our analysis reveals several structural properties of threshold algorithms, particularly that a threshold algorithm with a threshold of $T$ guarantees at least the minimum of $F(T)$ and $1-F(T)$, where $ F $ is the CDF of the maximum reward. 
Moreover, for a threshold $T$, if the expected number of rewards exceeding the threshold is $\delta$, then $ F(T) $ must lie within the interval $[1 - \delta, e^{-\delta}]$. Carefully combining these two observations with a concentration bound implies our guarantee. 
The optimality is followed by considering two instances depending on the threshold's ranking. 
If the ranking of the chosen threshold is too high, then the algorithm might miss a hidden high reward that appears with a very low probability. 
If the ranking is too low, then the probability of selecting a reward at all is too low when all rewards' values are approximately the same.

\paragraph{Simplified Analysis for the Single-Sample Case.} Our last result (Section~\ref{sec:simpler}) is a simpler proof of the main result of \citet{rubinstein2020optimal} where we show that for the single-sample case, the threshold algorithm that sets the maximum sample as a threshold achieves a competitive-ratio of $\frac{1}{2}$. We show this result by establishing that for each value $x$, the probability that the algorithm selects at least $x$ is at least half of the probability that the maximum is at least $x$. Thus, this algorithm not only achieves half of the expectation of the prophet, but also, the distribution of this algorithm $\frac{1}{2}$-stochastically dominates the distribution of the prophet. Our proof relies on establishing two lower bounds that compare the probability of selecting at least $x$ to the probability that the maximum is at least $x$.
We then integrate over our constructed lower bounds and derive the approximation of half. 

\subsection{Related Work}
The study of prophet inequalities with limited prior information was initiated by \citet{azar2014prophet}, who demonstrated that a single-sample from each distribution is sufficient to achieve an asymptotically optimal competitive-ratio for the $k$-capacity prophet inequality. 
\citet{rubinstein2020optimal} focused on the single-choice prophet inequality and 
showed that the optimal competitive-ratio of $\frac{1}{2}$ can be attained using just 
one sample per distribution. Additionally, they showed that when the rewards' values are drawn from 
identical distributions, $O(\frac{n}{\epsilon^6})$ samples suffice to achieve a loss of 
only an $\epsilon$-fraction compared to the best online algorithm with full 
distributional information. 
\citet{dutting2021prophet} explored the single-sample prophet inequality under a 
matching feasibility constraint and established a competitive-ratio of $\frac{1}{16}$. 
\citet{caramanis2022single} extended the single-sample prophet inequality framework 
to various settings, including different types of matroids and matching problems under 
diverse arrival models. 
\citet{gravin2022optimal} considered the prophet inequality with fewer than one sample 
per distribution, where each distribution is sampled with probability $p$, and 
demonstrated that the optimal competitive-ratio achievable in this setting is 
$\frac{p}{p+1}$. 
\citet{correa2024sample} improved the sample complexity results of \citet{rubinstein2020optimal} 
for the i.i.d. prophet inequality, showing that $\frac{n}{\epsilon}$ samples are sufficient to achieve a loss of an $\epsilon$-fraction. 
Recently, \citet{cristi2024prophet} introduced a unified framework showing that $O\left(\frac{1}{\epsilon^5}\right)$ of samples from all distributions is sufficient to achieve the optimal competitive-ratio (up to a factor of $\epsilon$) for three variants of the prophet inequality: 
prophet secretary (where the rewards arrive in random order), free order (where the 
decision-maker can choose the order of the observed rewards), and the i.i.d. case. 
\citet{fu2024sample} examined the prophet inequality under matroid feasibility constraints 
and proved that a polylogarithmic number of samples per distribution is enough to achieve 
a constant competitive-ratio.

Our work also connects to the concept of identity-blind algorithms introduced by 
\citet{ezra2024choosing}. An identity-blind algorithm is an algorithm that does not know which 
distribution each value is drawn from. In this paper, we consider algorithms that are 
identity-blind with respect to the identities of the samples.

\section{Preliminaries}
We consider a setting with $n$ boxes. Every box $i$ contains some value $v_i$ drawn from an underlying independent distribution $F_i$. We denote the product distribution $F = F_1\times \ldots \times F_n$.
The values are revealed sequentially in an online fashion and the decision-maker needs to select one of them in an immediate an irrevocable manner without observing future values.
We consider algorithms that do not know the distributions $F_i$'s, but instead will have access to a pool of samples from them. 
We assume that prior to the online decision-making process where $v_i$ are revealed to the algorithm, a multi-set of size $k\cdot n$ samples is given to the algorithm where the multi-set is composed of $k$ samples from each distribution. The algorithm is not aware of which sample is drawn from which distribution. 
We denote by $\SET{F}{k}$ the distribution of the multi-set given to the algorithm where $F$ is the product distribution and $k$ is the number of samples from each $F_i$. For a value $x $ in the supports of the distributions $F_1,\ldots,F_n$, and a multi-set $S$ we denote by $\oc{x}{S}$ the number of occurrences of $x$ in $S$. 

We denote an online algorithm by $\alg$, and given a multi-set $S$ of samples and a vector of values $v$, we denote by $\alg(S,v)$ the (possibly random) value chosen by the algorithm $\alg$ on $v$ and $S$. 

For a parameter $k$, we measure the performance of an algorithm $\alg$ that receives a pool of $k$ samples from each distribution by the competitive-ratio measure which is $$ \alpha_k(\alg)  := \inf_{F} \frac{\E_{S \sim \SET{F}{k} }\left[\E_{v\sim F} [\alg(S,v)]\right]}{\E_{v\sim F} [\max_i v_i]} ,$$
where the expectation in the numerator is also taken over the randomness of the algorithm.

When proving lower bounds on $\alpha_k$, we assume for simplicity that all distributions are continuous, therefore there are no ties between different values and samples. This assumption is made without loss of generality by employing standard techniques, such as using randomized tie-breaking when values or samples are identical (see e.g., \citep{ezra2018prophets,rubinstein2020optimal}).

\paragraph{Probability definitions and tools.}
Our proofs rely on the following metric of distributions:
\begin{definition}[Total Variation Distance]
For a measurable space $(\Omega, \mathcal{F})$, and two probability measures $D_1,D_2$ defined on $(\Omega, \mathcal{F})$, the total variation distance between $D_1$ and $D_2$ is defined by $$ \delta_{TV}(D_1,D_2) := \sup_{A\in\mathcal{F}} |D_1(A)-D_2(A)|.$$
\end{definition}
We use $\bin{n}{p}$ to denote the binomial distribution with parameters $n,p$, we use $\norm{\mu}{\sigma^2}$ to denote the normal distribution with an expectation of $\mu$, and variance $\sigma^2$, we use $\ber{p}$ to denote the Bernoulli distribution with a parameter $p$, and we use $U(a,b)$ to denote the uniform distribution on the interval $[a,b]$.
In particular, we use the following properties of total variation distance:
\begin{theorem}[\citep{gavalakis2024entropy}]\label{thm:discrete}
Fixing $p\in(0,1)$, and taking $n$ goes to infinity, the total variation distance between binomial with parameters $n,p$, and a normal distribution with the same mean and variance goes to $0$. Formally, $$\lim_{n\rightarrow\infty} \delta_{TV}\left(\bin{n}{p},\norm{np}{np(1-p)}\right) =0 ,$$
where  
  $$\delta_{TV}\left(\bin{n}{p},\norm{np}{np(1-p)}\right) = \sum_{i=-\infty }^ \infty |\Pr_{X\sim \bin{n}{p}}[X=i] - \Pr_{Y\sim \norm{np}{np(1-p)}}[i-0.5\leq Y\leq i+0.5]  |.$$
\end{theorem}

We also note the following bound on the total variation distance of two normal distributions:
\begin{theorem}[\citep{devroye2018total}]\label{thm:distance}
For two normal distributions with the same mean, it holds that:
$$
    \delta_{TV}\left(\norm{\mu}{ \sigma_1^2}, \norm{\mu}{ \sigma_2^2}\right) \leq O\left( \left|\frac{\sigma_1^2}{\sigma_2^2}-1\right|\right).
$$
\end{theorem}

We also use the following version of the Chernoff  bound:
\begin{theorem}[Chernoff  Bound]
For every $\delta \in (0,1)$, and $X_1,\ldots,X_n$  independent Bernoulli random variables where $X=\sum_i X_i$, with $\mu= \E[X]$ it holds that
   $$ \Pr[|X -\mu|\geq \delta \mu ] \leq 2e^{-\delta^2\mu/3}.$$
\end{theorem}
Lastly, we are comparing different distributions through the notion of approximate stochastic domination.
\begin{definition}[Stochastic Domination]
For two distributions over $\mathbb{R}_{\geq 0}$ and a parameter $\gamma \in [0,1]$ we say that $D_1$ $\gamma$-stochastically dominates the distribution $D_2$ if for every  $x\in \mathbb{R}_{\geq 0}$ it holds that $$\Pr_{X \sim D_1}[ X\geq x] \geq \gamma \cdot \Pr_{X \sim D_2}[ X\geq x] .$$
\end{definition}
In particular, in some of our proofs, we show that the distribution of the chosen value of an algorithm $\gamma$-stochatically dominates the distribution of the prophet which implies a competitive-ratio of $\gamma$ since 
\begin{equation}
    \E[\alg(v)] = \int_{0}^\infty \Pr[\alg(v) \geq x]dx  \geq  \int_{0}^\infty \gamma \cdot\Pr[\max_i v_i \geq x] dx   = \gamma \cdot\E[\max_i v_i] .\label{eq:stochastic-dominance}
\end{equation}
\section{Hardness of Approximation}
\label{sec:main}
In this section, we show that in contrast to the single-sample case, for large enough $k$, no algorithm can guarantee a competitive-ratio of $
\frac{1}{2}$.
\begin{theorem}\label{thm:hard}
    There exists $k_0$ and $ c < \frac{1}{2}$ such that for every $k \geq k_0$, and every algorithm $\alg$, it holds that: $$ \alpha_k(\alg) \leq c.$$
\end{theorem}

To prove this theorem, we introduce a family of distributions $F$'s where any algorithm must obtain strictly less than $c$ times the value of the expected maximum for one of the distributions in the family. We consider the case where $n=6$, and $c= 0.4997$. In our analysis, we use the following parameters to be set later $\xi$, $\delta_1$, $\delta_2$, $\epsilon \in (0,1)$. Let $u= (\xi,1,1,1,1,k^4)$. In our analysis, we treat terms that converge to $0$ as $k_0$ goes to infinity by $o(1)$.
Given a vector of probabilities $p\in [0,1]^n$, we define the distribution $F^p$ to be such that for every $i\in [n]$, $v_i =u_i$ with probability $p_i$, and $0$ otherwise.
We only consider vectors of probabilities where $p_1=1$, $p_2,p_3,p_4\in \{0,\frac{1}{3},1\}$, $p_5\in [0,2\epsilon]$ $p_6\in\{0,\frac{1}{k^3}\}$. We denote this support of vector of $p$'s by $\ps$.

We let the algorithm know that we only consider distributions from the family $\{F^p \mid p\in \ps\}$. This provides more information to the algorithm, and allows us to analyze the behavior of the algorithm in certain scenarios, and characterize all possible algorithms for the inputs derived from this family of distributions.
The multi-set of samples is then defined by $\oc{0}{S},\oc{\xi}{S},\oc{1}{S}, \oc{k^4}{S}$, where $\oc{\xi}{S}=k$, and $\oc{0}{S}+\oc{\xi}{S}+\oc{1}{S}+ \oc{k^4}{S} =n\cdot k$.

We note, that when the samples contain an occurrence of $k^4$ (i.e., $\oc{k^4}{S}>0$) then it must be that $p_6 =\frac{1}{k^3}$, and it is clear that the only optimal online strategy is to reject the first 5 values since the expectation of the last one is strictly larger than all the values in the support of $F_1,\ldots,F_5$. This strategy obtains an expected value of $k$. Thus, we will assume that if $k^4$ was observed in the samples, then the algorithm follows this strategy. We can also assume that the algorithm never selects a value of $0$. 
We denote the event that a value of $k^4$ was observed in the samples by $\event{}(S)$ where $S$ is the sampled set.
We also denote the event that there are $i$ occurrences of $1$ in the sampled multi-set by $\event{i}(S)$, where $0 \leq i \leq 4k$.
We denote by $\pre$ the set of all observable prefixes of $v$ of length at most $5$, that is $$\pre:=  \{\xi\} \cup \left(\{\xi\} \times \{0,1\}\right) \cup \left(\{\xi\} \times \{0,1\}^2\right) \cup \left(\{\xi\} \times \{0,1\}^3 \right) \cup \left(\{\xi\} \times \{0,1\}^4\right).$$  
Thus, an online algorithm that satisfies our assumptions can be characterized by a mapping $q:\pre\times \{0,\ldots,4k\} \rightarrow [0,1]$, which defines for each prefix $t\in\pre$, and a number of occurrences of $1$'s in the samples $i\in \{0,\ldots,4k\}$ the probability $q(t,i)$ of selecting the last value of the prefix $t$ conditioned on $\event{i}(S)\wedge \bar{\event{}}(S)$ and that the last value was reached.
For a vector $p\in \ps$ and  a prefix $t\in\pre$, we denote by $$q_p(t) :=\E_{S\sim \SET{F^p}{k}}[q(t,i) \mid \event{i}(S) \wedge \bar{\event{}}(S)].$$
We first analyze the probability that $\alg$ selects the value $\xi$. We define the following prefixes: $t_1:=(\xi)$, $t_2:= (\xi,1)$, $t_3:=(\xi,0,1)$, and $t_4:=(\xi,0,0,1)$.

\begin{lemma}\label{lem:1}
    If there exists $\ell\in [0,2\epsilon]$ such that for the distribution $F^p$ with $p= (1,1,0,0,\ell,0)$ it holds that $\Pr[q_p(t_1) ] \geq \delta_1$, then there exists  $p'\in\{(1,1,0,0,\ell,0),(1,1,0,0,\ell,\frac{1}{k^3})\}$ such that  $$\frac{\E_{S \sim \SET{F^{p'}}{k} }\left[\E_{v\sim F^{p'}} [\alg(S,v)]\right]}{\E_{v\sim F^{p'}} [\max_i v_i]}\leq c .$$
\end{lemma} 
\begin{proof}
We say that the algorithm over-select for some $\ell$ if: $$\sum_{i=0}^{4k} \Pr[\event{i}(S)] \left(q(t_1,i) + (1-q(t_1,i)) q(t_2,i)\right) \geq \delta_2.$$
In other words, the expected probability after separating into different cases depending on $\oc{1}{S}$ that the algorithm selects a value among $v_1,v_2$, given that $v$ starts with $(\xi,1)$ is at least $\delta_2$.
\paragraph{Case 1: The algorithm over-select for some $\ell\in[0,2\epsilon]$.}  Then for $p'= (1,1,0,0,\ell,\frac{1}{k^3})$ it holds that \begin{equation}
    \E_{v\sim F^{p'}} [\max_i v_i]  \geq  k. \label{eq:opt1}
\end{equation}
However, if the algorithm selects a value among $v_1,\ldots,v_5$, it gets at most a value of $1$, and if not, it gets at most the expectation of $v_6$ which is $k$. Thus,
\begin{eqnarray}
    \E_{S \sim \SET{F^{p'}}{k} }\left[\E_{v\sim F^{p'}} [\alg(S,v)]\right] & \leq & 1 + k\cdot \Pr[\alg \mbox{ reaches }v_6].  \label{eq:alg1}  
\end{eqnarray}
Moreover, in order to reach $v_6$, it must not select a value among $v_1,v_2$, which allows us to bound the probability by the following expression
\begin{eqnarray}    
    \Pr[\alg \mbox{ reaches }v_6]  & \leq &  1 -   \sum_{i=0}^{4k} \Pr[\event{i}(S) \wedge \bar{\event{}}(S)] \left(q(t_1,i) + (1-q(t_1,i)) q(t_2,i)\right) \nonumber \\ & = &  1 -  \Pr[\bar{\event{}}(S)] \cdot  \sum_{i=0}^{4k} \Pr[\event{i}(S) ] \left(q(t_1,i) + (1-q(t_1,i)) q(t_2,i)\right)  \nonumber \\ &\leq & 1 -  \Pr[\bar{\event{}}(S)] \cdot \delta_2 = 1- \delta_2 +o(1), \label{eq:reach5}
\end{eqnarray}
where the first inequality is by bounding the probability that we select a value among $v_1,v_2$; the first equality is since $\event{}(S)$ and $\event{i}(S)$ are independent; the second inequality is by the definition of over-selection at $\ell$; the last equality is since $\Pr[\event{}(S)] \leq \frac{1}{k^2}$.

By combining Equations~\eqref{eq:opt1}, \eqref{eq:alg1}, \eqref{eq:reach5} we get that $$\alpha_{k}(\alg) \leq \frac{1+k(1-\delta_2+o(1))}{k} = 1-\delta_2 +o(1) .$$
If our parameters satisfy that 
\begin{equation}
    1-\delta_2<c, \label{eq:cond1}
\end{equation}
then this concludes the proof of this case.
\paragraph{Case 2: The algorithm does not over-select for any $\ell\in[0,2\epsilon]$.} For $p'=p$, it holds that 
\begin{equation}
    \E_{v\sim F^{p'}} [\max_i v_i]  = 1. \label{eq:opt2}
\end{equation}
However, the probability that $\alg$ selects $\xi$ is at least 
\begin{eqnarray}
    \E_{S \sim \SET{F^{p'}}{k} }\left[\E_{v\sim F^{p'}} [\alg(S,v)]\right]  =  \xi \cdot \Pr[\alg \mbox{ selects } \xi] +   \Pr[\alg \mbox{ selects } 1] \leq \xi\cdot \delta_1 + (\delta_2-\delta_1 +2\epsilon) ,  \label{eq:alg2}
    \end{eqnarray}
    where the last inequality is since by the assumption of the claim, the probability of selecting $\xi$ is at least $\delta_1$, and the probability of selecting something among the first two values is at most $\delta_2$, and the probability of observing a non-zero value in $v_3,v_4,v_5,v_6$ is bounded by $2\epsilon$.

If we choose parameters that satisfy \begin{equation}
    \xi\cdot \delta_1+\delta_2-\delta_1 +2\epsilon < c ,\label{eq:cond2}
\end{equation}
    then by combining with Equations~\eqref{eq:opt2} and \eqref{eq:alg2} it concludes the proof. 
\end{proof}
We next claim that the algorithm cannot over-select  the prefixes $t_2,t_3,t_4$. The following claim follows from the analysis of the first case of Claim~\ref{lem:1}.  
\begin{lemma}\label{lem:2}
Let $p^2= (1,1,0,0,\ell,0)$, $p^3= (1,0,1,0,\ell,0)$, $p^4= (1,0,0,1,\ell,0)$.
    For $j\in\{2,3,4\}$, if there exists $\ell\in [0,2\epsilon]$ such that for the distribution $F^{p^j}$ it holds that $\Pr[q_{p^j}(t_j) ] \geq \delta_2$, then for  $p'=p^j+ (0,0,0,0,0,\frac{1}{k^3})$ it holds that  $$\frac{\E_{S \sim \SET{F^{p'}}{k} }\left[\E_{v\sim F^{p'}} [\alg(S,v)]\right]}{\E_{v\sim F^{p'}} [\max_i v_i]}\leq c .$$
\end{lemma}

We next assume that $\alg$ does not satisfy the cases considered in Lemmas~\ref{lem:1} and \ref{lem:2}.
Consider the distribution $F^{p^*}$ for $p^*= (1,\frac{1}{3},\frac{1}{3},\frac{1}{3},\epsilon,0)$. 
We next bound the difference in $q$ for different distributions over the samples.
\begin{lemma}\label{lem:dist}
    For  $D_1,D_2$ two distributions over $\{0,\ldots,4k\}$, and $t \in \pre$, it holds that  
$$\E_{i\sim D_1}[q(t,i) ] \leq \E_{i\sim D_2}[q(t,i) ] + \delta_{TV}(D_1,D_2)  .$$
\end{lemma} 
\begin{proof}
    It holds that 
    \begin{eqnarray*}        
  \E_{i\sim D_1}[q(t,i) ]  & = & \sum_{i=0}^{4k} q(t,i) \cdot \Pr_{X\sim D_1} [X=i]  \\ &\leq & \sum_{i=0}^{4k} q(t,i) \cdot \left(\Pr_{Y\sim D_2} [Y=i]+ \left( \Pr_{X\sim D_1} [X=i]-\Pr_{Y\sim D_2} [Y=i]\right)^+\right)   \\ &\leq &     
  \E_{i\sim D_2}[q(t,i) ]  + \delta_{TV}(D_1,D_2) , 
      \end{eqnarray*}
      where the last inequality is since $q(t,i) \leq 1$ and since $\sum_{i=0}^{4k}  \left( \Pr_{X\sim D_1} [X=i]-\Pr_{Y\sim D_2} [Y=i]\right)^+ = \delta_{TV}(D_1,D_2) $.
\end{proof}

We next extend Lemma~\ref{lem:dist} to convex combinations of distributions:
\begin{lemma}\label{lem:dist2}
    Let  $D,D_1,D_2,\ldots,D_m$ be distributions over $\{0,\ldots,4k\}$, and let $c_1,\ldots,c_m \in [0,1]$ be coefficients that sum up to $1$. Let $D'$ be the distribution over $\{0,,\ldots 4k\}$ constructed by first drawing $j$ with probability $c_j$, and then drawing $i$ from $D_j$. Then for every $t \in \pre$, it holds that  
$$\E_{i\sim D}[q(t,i) ] \leq \E_{i\sim D'}[q(t,i) ] + \delta_{TV}(D,D')  .$$
\end{lemma} 
\begin{proof}
    It holds that 
    \begin{eqnarray*}        
  \E_{i\sim D}[q(t,i) ]  & = & \sum_{i=0}^{4k} q(t,i) \cdot \Pr_{X\sim D} [X=i]  \\   & = & \sum_{i=0}^{4k} q(t,i) \sum_{j=1}^m c_j \Pr_{X\sim D} [X=i] \\ & = & \sum_{i=0}^{4k} q(t,i) \sum_{j=1}^m c_j  \left(\Pr_{Y_j\sim D_j} [Y_j=i]+ \left(  \Pr_{X\sim D} [X=i]-\Pr_{Y_j\sim D_j} [Y_j=i]\right)\right) \\& = & \sum_{i=0}^{4k} q(t,i) \left(\Pr_{Y\sim D'}[Y=i] + \left(\Pr_{X\sim D}[X=i] -\Pr_{Y\sim D'}[Y=i]\right)\right)  \\& \leq & \sum_{i=0}^{4k} q(t,i) \left(\Pr_{Y\sim D'}[Y=i] + \left(\Pr_{X\sim D}[X=i] -\Pr_{Y\sim D'}[Y=i]\right)^+\right)  \\ &\leq & 
  \E_{i\sim D'}[q(t,i) ]  + \delta_{TV}(D,D') , 
      \end{eqnarray*}
      where the last inequality is since $q(t,i) \leq 1$ and since $\sum_{i=0}^{4k}  \left( \Pr_{X\sim D_1} [X=i]-\Pr_{Y\sim D_2} [Y=i]\right)^+ = \delta_{TV}(D,D') $.
\end{proof}

So far, in Lemmas~\ref{lem:1} and \ref{lem:2} we showed that for some distributions on the sampled $\SET{F^p}{k}$ imposed by an $F^p$ that satisfies the conditions of the lemmas, the algorithm cannot select $\xi$ with a probability significantly larger than $\delta_1$, and cannot select a value of $1$ under prefixes $t_2,t_3,t_4$ with a probability significantly more than $\delta_2$.   
Our next step is to utilize Lemma~\ref{lem:dist} to show that the same holds for $F^{p^*}$.
The main difficulty is that the distribution over the samples $\SET{F^{p^*}}{k}$ imposed by $F^{p^*}$ is far in terms of total variation distance from any distribution over samples imposed by $F^{p}$ that satisfies Lemma~\ref{lem:1} or Lemma~\ref{lem:2}.
However, Lemma~\ref{lem:dist2} allows us to consider convex combinations of such distributions. In particular, we show that the distribution over $\oc{1}{S}$ where $S\sim \SET{F^{p^*}}{k}$ is close enough to a convex combination of $\oc{1}{S}$ where $S\sim\SET{F^p}{k}$ of $p$'s that satisfy one of the lemmas.

\begin{lemma}\label{lem:dd}
Let $D^*$ be the distribution of $\oc{1}{S}$ where $S\sim\SET{F^{p^*}}{k}$.
    There exist coefficients  $c_1,\ldots c_m\in [0,1]$ that sum up to $1$ and vectors $p^1,\ldots,p^m$ that satisfies the conditions of Lemma~\ref{lem:1} and \ref{lem:2} such that the distribution $D$ over $\{0,\ldots,4k\}$ where $j$ is drawn with probability $c_j$, and then $\oc{1}{S}$ is sampled where $S\sim F^{p^j}$ satisfies $$ \delta_{TV}(D,D^*) \leq  o(1).$$  
\end{lemma}
\begin{proof}
    
We first observe that $\oc{1}{S}$ where $S\sim \SET{F^{p^*}}{k}$ is a sum of two binomials. Let $X \sim \bin{3k}{1/3}$ and let $Y\sim \bin{k}{\epsilon} $. An alternative way to define the distribution $D^*$ of $X+Y$, is to define a conditional variable $Y_x = Y+x$, and consider the distribution of $Y_X$.

On the other hand, consider the following convex combination of the distributions: We think of the coefficients of each $p^j$ as probabilities, where $p^j$ is sampled with a probability of $\Pr[\bin{3k}{1/3}=j]$.  
We denote the function $g:\{0,\ldots,3k\}\rightarrow[0,1]$ where $g(x):=\min(1,\max(0,\frac{x-k}{k}+\epsilon))$. Then let $p^{j}=(1,1,0,0,g(j),0)$.
We first draw $X'$ from $\bin{3k}{\frac{1}{3}}$, then the amount $ \oc{1}{S}$ where $S\sim \SET{F^{p^{X'}}}{k}$ is distributed as 
 $Z_{X'} = k+ \bin{k}{\epsilon+g(X')}$.  This is equivalent to choosing $p^j$ with a probability of $\Pr[X'=j]$.

We bound the $\delta_{TV}(D,D^*)$.
It holds that 
\begin{eqnarray}
    \delta_{TV}(D,D^*) & = & \sum_{i=0}^{4k} |\Pr[Y_X=i]-\Pr[Z_{X'}=i]| \nonumber
\\ & \leq &
\sum_{x=0}^{3k}\sum_{i=0}^{4k} |\Pr[Y_X=i \wedge X=x]-\Pr[Z_{X'}=i \wedge X'=x]|
\nonumber \\ &=& \sum_{x=0}^{3k}\Pr[X=x]\sum_{i=0}^{4k} |\Pr[Y_x=i]-\Pr[Z_{x}=i]|
\nonumber \\ & \leq &  o(1) + \sum_{x=k-k^{2/3}}^{k+k^{2/3}}\Pr[X=x]\sum_{i=0}^{4k} |\Pr[Y_x=i]-\Pr[Z_{x}=i]|, \label{eq:bound}
\end{eqnarray}
where the first inequality is by refining the outcome space; the second inequality is since the probability of binomial to be $\Theta(k^{1/6})$ standard deviations from its mean is negligible as $k$ grows large.

We also note that $Y_x$ is a constant ($x$) plus a binomial distribution ($\bin{k}{\epsilon}$), so it has an expectation of $\mu_Y=x+k\epsilon$, and a variance of $\sigma_Y^2=k\epsilon(1-\epsilon)$. By Theorem~\ref{thm:discrete} it holds that\footnote{Here we also use that adding a constant to the two random variables does not change their distance.} \begin{equation}
\lim_{k\rightarrow \infty }    \delta_{TV}(D^*,\norm{\mu_Y}{\sigma_Y^2}) =0 . \label{eq:dis1}
\end{equation}  

On the other hand, $Z_{x}\sim k+ \bin{k}{\epsilon+g(x)}$, which has an expectation of $\mu_Z= k+ kg(x) = k + (\epsilon+\frac{x-k}{k})k = x+k\epsilon$, and has a variance of 
$\sigma_Z^2=k(\epsilon+g(x))(1-\epsilon-g(x))  =k(\epsilon-\epsilon^2\pm o(1)) $, which gives us by Theorem~\ref{thm:discrete} that 
\begin{equation}
\lim_{k\rightarrow\infty}\delta_{TV}(D,\norm{\mu_Z}{\sigma_Z^2}) =0 \label{eq:dis2}.
\end{equation}

We also know by Theorem~\ref{thm:distance}  since $\mu_Y=\mu_Z$, and by possible range of $\sigma_Z^2$ that 
\begin{eqnarray}
    \delta_{TV}(\norm{\mu_Y}{\sigma_Y^2},\norm{\mu_Z}{\sigma_Z^2}) & \leq &  O\left(\left| \frac{\sigma_Y^2}{\sigma_Z^2}-1\right| \right) \nonumber \\  &  \leq & O\left(\max\left\{\left|\frac{k\epsilon(1-\epsilon)}{k(\epsilon-\epsilon^2+o(1))} -1 \right|, \left| \frac{k\epsilon(1-\epsilon)}{k(\epsilon-\epsilon^2-o(1))} -1\right|\right\}\right) =o(1). \label{eq:disnorm}
\end{eqnarray}

Combining Equations~\eqref{eq:dis1}, \eqref{eq:dis2} and \eqref{eq:disnorm} and by the triangle inequality of $\delta_{TV}$ we get that 
$$ \sum_{i=0}^{4k} |\Pr[Y_x=i]-\Pr[Z_{g(x)}=i]| =o(1),$$
which by combining with Equation~\eqref{eq:bound} we get that $$   \delta_{TV}(D,D^*) \leq  o(1) + \sum_{x=k-k^{2/3}}^{k+k^{2/3}}\Pr[X=x] \cdot o(1) = o(1), $$
which concludes the proof.
\end{proof}

\begin{proof}[Proof of Theorem~\ref{thm:hard}]

By Lemmas~\ref{lem:1}, \ref{lem:dist2} and \ref{lem:dd}, $q_{p^*}(t_1) \leq \delta_1+o(1)$, and for $i=2,3,4$, by Lemmas~\ref{lem:2}, \ref{lem:dist2} and \ref{lem:dd}, it holds that $q_{p^*}(t_i) \leq \delta_2+o(1)$. Thus, for every $i$, after not selecting a value on $t_i$, if no non-zero value appears, the algorithm, selects nothing. Let $h(v)= v_2+v_3+v_4$. We bound the performance of the algorithm by partitioning into three types of online inputs depending on the values of $h(v)$: 1) $h(v)=0$ which happens with probability $\frac{8}{27}$ for which the algorithm can get at most $\xi\delta_1 +\epsilon$; 2) $h(v)=1$ which happens with probability $\frac{12}{27}$ for which the algorithm selects a value among them with probability at most $\delta_2$ and the algorithm gets at most $\delta_2+\epsilon$; 3) $h(v)>1$ which happens with a probability of $\frac{7}{27}$, in which the algorithm gets at most $1$. 
Thus, 
 \begin{eqnarray*}
     \E_{S \sim \SET{F^{p^*}}{k} }\left[\E_{v\sim F^{p^*}} [\alg(S,v)]\right]  & = & 
      \E[\alg(S,v) \mid h(v)=0] \cdot \Pr[h(v)=0] \\ & + & 
      \E[\alg(S,v) \mid h(v)=1] \cdot \Pr[h(v)=1] \\ & + & \E[\alg(S,v) \mid h(v)>1] \cdot \Pr[ h(v)>1]  \\  & \leq &
      (\xi \delta_1 +\epsilon)\cdot \frac{8}{27} + (\delta_2+\epsilon) \cdot \frac{12}{27} + \frac{7}{27} +o(1). 
 \end{eqnarray*}

 Overall, by setting $\xi=0.9,\delta_1=0.01,\delta_2=0.5005,\epsilon = 0.0001$ we get that for $$c > 0.4997 =\max\left(\xi\delta_1+\delta_2-\delta_1+2\epsilon,1-\delta_2,\xi\cdot \delta_1 \cdot \frac{8}{27} +  \delta_2 \cdot \frac{12}{27} + \frac{7}{27}+\epsilon\right), $$
no algorithm can guarantee a $\alpha_k(\alg)\geq c$ which concludes the proof.
\end{proof}

\section{Ordinal Static Threshold Algorithms}
\label{sec:ordinal}
In this section, we provide an algorithm that guarantees a constant competitive-ratio.
In particular, we consider the class of ordinal static threshold algorithms where the algorithm selects the $i$-th highest sample (independent of the samples' values) as a threshold, and then selects the first value that exceeds it.
We know that for the case where $k=1$, the single-threshold algorithm that selects the maximum sample, i.e., selecting $i=1$ (analyzed by \citet{rubinstein2020optimal}) provides a competitive-ratio of $\alpha_1(\alg) = \frac{1}{2}$.
Let $F_i$  be the CDF of box $i$. Let $F=\prod_{i} F_i$, be the product of the CDF's (which is the CDF of the maximum).

We next show that when $k \rightarrow \infty$, the optimal ordinal static threshold algorithm, uses the $\ell \approx \rho \cdot k$ highest sample as a threshold, and achieves a competitive-ratio of $\lim_{k\rightarrow\infty } \alpha_k(\alg) = 1-\rho$ where $\rho\approx 0.567 $ is the solution to the equation \begin{equation}
    xe^x=1. \label{eq:rho}
\end{equation}
\begin{theorem}\label{thm:ob-pos}
    The single-threshold algorithm $\alg$ that selects the $\ell$-th highest sample as a threshold for $\ell=\rho \cdot k -k^{2/3}$ guarantees a competitive-ratio  of $$\alpha_k(\alg) \geq 1-\rho -o(1).$$
\end{theorem}
\begin{proof}
Let $\alg_T$ be the single-threshold algorithm with threshold $T$.
Our proof is divided into three claims:
\begin{enumerate}
    \item Algorithm $\alg_T$ $h(T)$-stochastically dominates the prophet for  $h(T):= \min\{F(T),1-F(T)\}$.\label{item:1} 
    \item For $g(T):= \sum_{i}\Pr[v_i> T] $, it holds that $  1- g(T) \leq F(T)   \leq e^{-g(T)}  $. \label{item:2}
    \item Almost surely $g(T)=\rho\pm o(1)$. \label{item:3}
\end{enumerate}
Combining Claims~\ref{item:1}, \ref{item:2}, \ref{item:3} with Equation~\eqref{eq:stochastic-dominance} concludes the proof since $g(T) = \rho\pm o(1)$ implies that $1-\rho-o(1) \leq 1-g(T) \leq F(T)\leq e^{-g(T)} =\rho+o(1)$ which implies that $h(T)\geq  1-\rho-o(1)$.

\paragraph{Proof of Claim~\ref{item:1}.} 
We need to show that for every value $x$, it holds that: 
 \begin{equation}
    \Pr[\alg_T(v) \geq x] \geq h(T) \cdot \Pr[\max_i v_i \geq x]. \label{eq:stochastic-gen}
\end{equation}
To show Inequality~\eqref{eq:stochastic-gen}, we consider two cases depending on whether $x\leq T$, or $x>T$.
If $x\leq T$ then $$ \Pr[\alg_T(v)\geq x] = [\alg_T(v)\geq T] = 1-F(T)\geq (1-F(T))\cdot\Pr[\max_i v_i \geq x].$$
If $x>T$
then \begin{equation*} 
\Pr[\alg_T(v)\geq x]   =   \sum_{i}\Pr[v_i\geq x]\prod_{j<i} \Pr[v_i\leq T] \geq \sum_{i}\Pr[v_i\geq x]\cdot F(T) \geq F(T)\cdot\Pr[\max_i v_i \geq x],\end{equation*}
which concludes the proof of the claim.

\paragraph{Proof of Claim~\ref{item:2}.}
The left side of the inequality is equivalent to $$ 1-\prod_i F_i(T)\leq \sum_i (1-F_i(T),$$
which holds by the union bound.

The right side of the inequality is equivalent to $$ \prod_i F_i(T) \leq \prod_i e^{-(1-F_i(T))},$$
which holds by applying to each $F_i(T)$ the inequality $ z \leq e^{-(1-z)}$  that holds for all $z\in [0,1]$.

\paragraph{Proof of Claim~\ref{item:3}.}
In order for $g(T) >\rho+k^{-1/3}$, there should be at most $\rho k$ samples above $T^+= g^{-1}(\rho+k^{-1/3}) $. If we denote by $q_i^+=\Pr[v_i > T_+]$, then this is equivalent to $$ \Pr[\sum_{i} \bin{k}{q_i^+} \leq \rho k]  \leq  2e^{-k^{1/3}/3} =o(1) .$$
Similarly, in order for $g(T) <\rho+k^{-1/3}$, there should be at least $\rho k$ samples above $T^-= g^{-1}(\rho-k^{-1/3}) $. If we denote by $q_i^-=\Pr[v_i > T_-]$, then this is equivalent to $$ \Pr[\sum_{i} \bin{k}{q_i^-} \geq \rho k]  \leq  2e^{-k^{1/3}/3} =o(1) .$$
This concludes the proof.
\end{proof}

We complement Theorem~\ref{thm:ob-pos} by showing that no ordinal static threshold algorithm guarantees a better competitive-ratio.
\begin{theorem}\label{thm:ob-neg}
    For every $\ell\in\{1,\ldots,kn\}$ the single-threshold algorithm $\alg$ that selects the $\ell$-th highest sample as a threshold satisfies that  $$\alpha_k(\alg) \leq 1-\rho +o(1).$$
\end{theorem}
\begin{proof}
    We consider two cases:
\paragraph{Case 1:}    If $\ell\geq \rho k$, then consider the scenarios where the value of the first box is $1 + U(0,1)$, and the second box has a value of $k^3\cdot  \ber{\frac{1}{k^2}} + U(0,1) $.  It holds that $\E[\max_i v_i] \geq \E[v_2] \geq k $.
    On the other hand $$\Pr[T(S)>2-\rho+k^{-1/3}] =\Pr[\bin{k}{\rho-k^{-1/3} } +\bin{k}{\frac{1}{k^2}} \geq \ell] = o(1) .$$
    Thus, 
    \begin{eqnarray}
    \E[\alg(S,v)] & = & \Pr[T(S)>2-\rho+k^{-1/3}] \cdot \E[\alg(S,v) \mid T(S)  >2-\rho+k^{-1/3}] \nonumber\\  & +&   \Pr[T(S)\leq 2-\rho+k^{-1/3}] \cdot \E[\alg(S,v) \mid T(S)  \leq 2-\rho+k^{-1/3}] \nonumber \\ & \leq & o(1) \cdot (k+1) + (1-o(1)) \cdot (2+(1-\rho+k^{-1/3})(k+1)) = k(1-\rho+o(1)), \nonumber
    \end{eqnarray}
    where the inequality is since the expected value of the algorithm is at most the expected value of the optimal online algorithm that always selects $v_2$ that achieves an expected value of at most $k+1$, and since if the chosen threshold is at most $2-\rho+k^{-1/3}$ then the algorithm achieves at most the value of $v_1$ which is bounded by $2$ plus the expected value of $v_2$ times the probability that $v_2$ is reached which is at most $1-\rho+k^{-1/3}$.

\paragraph{Case 2:}  If $\ell\leq \rho k$, then consider the scenario where there are $n$ boxes where all of their values are distributed as $k + U(0,1)$. We assume that $n\rightarrow\infty$, but $n < k^{1/4}$.
 It holds that $\E[\max_i v_i] \geq k $.
    On the other hand $$ \Pr[T(S)  < k+ 1  -\frac{ \rho+k^{-1/3}}{n} ] = \Pr[\bin{kn}{ \frac{\rho+ k^{-1/3}}{n}} \leq \ell ] = o(1) .$$
Thus, 
    \begin{eqnarray}
    \E[\alg(S,v)] & = & \Pr[T(S)<k+ 1 - \frac{\rho+k^{-1/3}}{n}] \cdot \E[\alg(S,v) \mid T(S)  <k+ 1 - \frac{\rho+k^{-1/3}}{n}] \nonumber\\  & +&   \Pr[T(S)\geq k+ 1- \frac{\rho+ k^{-1/3}}{n}] \cdot \E[\alg(S,v) \mid T(S)  \geq k+ 1- \frac{\rho+k^{-1/3}}{n}] \nonumber \\ & \leq & o(1) \cdot (k+1) + (1-o(1)) \cdot \Pr[\bin{n}{ \frac{\rho+k^{-1/3}}{n}} >0] \cdot (k+1) \nonumber \\ & = & k\cdot (1-\Pr[\bin{n}{ \frac{\rho +k^{-1/3}}{n}} =0]) +o(1) = k\cdot (1-\left(1-\frac{\rho+k^{-1/3}}{n}\right)^{n}) +o(1) \nonumber \\ & = & k (1-e^{-\rho+o(1)}) +o(1)= k(1-\rho+o(1)), \nonumber
    \end{eqnarray}
        where the inequality is since the expected value of the algorithm is at most  $(k+1)$ times the probability of selecting some value, and since if the chosen threshold is at least $k+1-\frac{\rho+k^{-1/3}}{n}$ then the probability that $v_i>T$ is bounded by $\frac{\rho+k^{-1/3}}{n}$.
\end{proof}
\newpage
\section{Simpler Analysis for Single-Sample Prophet Inequality}
\label{sec:simpler}
In this section, we provide an alternative proof for the result of \cite{rubinstein2020optimal} that analyzes the case where $k=1$, i.e., when there is only one sample from each distribution. 
Consider the algorithm $\alg$ that sets the maximum sample as a threshold.
Let $F_i$  be the CDF of box $i$. Let $F=\prod_{i} F_i$, be the product of the CDF's (which is the CDF of the maximum), and let $f$ be the corresponding PDF of the maximum.
We show that $\alg$ $\frac{1}{2}$-stochastically dominates the prophet, i.e., for every value $x \in \mathbb{R}_{\geq 0}$, it holds that \begin{equation}
    \Pr[\alg(S,v) \geq x] \geq \frac{1}{2}\Pr[\max_i v_i \geq x]. \label{eq:stochastic}
\end{equation} 
Equation~\eqref{eq:stochastic} together with Equation~\eqref{eq:stochastic-dominance} imply that $\alpha_1(\alg) \geq \frac{1}{2}$. 
\begin{proof}[\textbf{Proof of Inequality~\eqref{eq:stochastic}}] We first observe that:
\begin{equation}
    \Pr[\max_i v_i \geq x] = 1-F(x).\label{eq:opt-single}
\end{equation} 
On the other hand, if the maximum sample which we denote by $s$, satisfies that   $s\leq x$, then:
\begin{eqnarray}
     \Pr[\alg(S,v)\geq x \mid \max(S)=s] & =& \sum_{i=1}^n (1-F_i(x))
\underbrace{\left(\prod_{j<i}F_{j}(s)\right)}_{\text{Reaching } i}  \nonumber\\ & = &
\sum_{i=1}^n (1-F_i(x))
\left(\prod_{j<i}F_{j}(x) \cdot \frac{F_j(s)}{F_j(x)}\right)  \nonumber\\ & \geq &
\frac{F(s)}{F(x)}\sum_{i=1}^n (1-F_i(x)) 
\left(\prod_{j<i}F_{j}(x) \right)    = \frac{F(s)}{F(x)} (1-F(x)),
\label{eq:b1}
\end{eqnarray}
where the inequality is since we multiplied each term by $\prod_{j \geq i}\frac{F_j(s)}{F_j(x)}$ which is at most $1$ since $s\leq x$.

If the maximum sample satisfies that $s>x$, then:
\begin{eqnarray}
\Pr[\alg(S,v)\geq x \mid \max(S)=s] & =& \sum_{i=1}^n (1-F_i(s))
\left(\prod_{j<i}F_{j}(s)\right)  =  (1-F(s))\label{eq:b2}.
\end{eqnarray}
Overall,
\begin{eqnarray}
    \Pr[\alg(S,v) \geq x] & = & \int_{0}^\infty f(s) \cdot    \Pr[\alg(S,v)\geq x \mid \max(S)=s] ds  
    \nonumber \\ & \geq & \int_{0}^x f(s) \cdot \frac{F(s)}{F(x)}(1-F(x)) ds +   \int_{x}^\infty f(s) (1-F(s))ds 
    \nonumber \\ &=& \frac{1-F(x)}{F(x)} \left[ \frac{1}{2}F(s)^2\right]_{0}^x + \left[ F(s)-\frac{1}{2}F(s)^2\right]_{x}^\infty \nonumber \\ & = & \frac{1-F(x)}{F(x)}\cdot  \frac{F(x)^2}{2} + 1-\frac{1}{2}-F(x) + \frac{1}{2}F(x)^2 = \frac{1-F(x)}{2}, \nonumber
\end{eqnarray}
where the first equality is by integrating over the maximum sample; the first inequality is by partitioning to values of $s$ that are smaller and larger than $x$ and by Equations~\eqref{eq:b1} and \eqref{eq:b2}; the second equality is since $\int f(s) ds =F(s) $ and  $\int f(s) F(s) ds  = \frac{1}{2}F(s)^2$.

Combining with Inequality~\eqref{eq:opt-single} concludes the proof.
\end{proof}
\bibliographystyle{abbrvnat}
\bibliography{bib}

\end{document}